\documentclass[onecolumn,11pt]{IEEEtran}

\usepackage{amsmath}
\usepackage{amsfonts,amssymb}
\usepackage{multirow}
  \usepackage[pdftex]{graphicx}
\usepackage{color}

\def\bX{\mathbf{X}}
\def\bY{\mathbf{Y}}

\def\bx{\mathbf{x}}

\def\EE{\mathbb{E}}

\def\FF{\mathbb{F}}
\def\bF{{\mathbb{F}}}

\def\CC{\mathbb{C}}
\def\RR{\mathbb{R}}

\def\SMS{\mathop{{\rm SMS}}\nolimits}

\newtheorem{theorem}{Theorem}
\newtheorem{remark}{Remark}
\newtheorem{lemma}{Lemma}

 \newenvironment{proof}{\par \noindent
            {\it Proof. \hspace{2mm}}}{\hfill$\Box$ \vspace*{3mm}}

\def\Label{\label}

\begin{document}
\title{Secure Modulo Sum via Multiple Access Channel
\thanks{The work reported here was supported in part by 
Guangdong Provincial Key Laboratory (Grant No. 2019B121203002),
the JSPS Grant-in-Aid for Scientific Research (A) No.17H01280, (B) No. 16KT0017,
and
Kayamori Foundation of Informational Science Advancement.}
}

\author{
        Masahito~Hayashi,~\IEEEmembership{Fellow,~IEEE}
\thanks{Masahito Hayashi is with 
Shenzhen Institute for Quantum Science and Engineering, Southern University of Science and Technology,
Shenzhen, 518055, China,
Guangdong Provincial Key Laboratory of Quantum Science and Engineering,
Southern University of Science and Technology, Shenzhen 518055, China,
Shenzhen Key Laboratory of Quantum Science and Engineering, Southern
University of Science and Technology, Shenzhen 518055, China,
and
the Graduate School of Mathematics, Nagoya University, Nagoya, 464-8602, Japan
(e-mail:hayashi@sustech.edu.cn).}}

\maketitle

\begin{abstract}
We discuss secure computation of modular sum 
when multiple access channel from 
distinct players $A_1, \ldots, A_c$ to a third party (Receiver) is given.
Then, we define the secure modulo sum capacity as the supremum of the transmission rate of 
modulo sum without information leakage of other information.
We derive its useful lower bound, which is numerically calculated under a realistic model that can be realizable as a Gaussian multiple access channel (MAC).
\end{abstract}


\section{Introduction}
Recently, secure multiparty computation has been actively studied even from the information theoretical viewpoint.
One of its simple examples is secure computation of modular sum \cite{CK,CS}.
This problem has been often discussed from the Shannon-theoretic viewpoint \cite{LA,BKMP}.
If we assume secure communication channels between a part of players,
secure computation of modular sum is possible.
Secure communication channels consumes cryptographic resources.
In this paper, instead of such resources, we focus on 
multiple access channel (MAC), which has been actively studied for a long time \cite{Shannon,Ahlswede,Liao}.
That is, we study how to realize secure computation via MAC instead of conventional secure communication channels.
As a typical example, we investigate secure computation of modular sum when a MAC from 
distinct players $A_1, \ldots, A_c$ to a third party (Receiver) is given.
That is, each player $A_i$ sends a 
sequence of their secure random number $M_i=(M_{i,1} , \ldots, M_{i,k})\in \FF_q^k$
that is subject to the uniform distribution,
and Receiver recovers only the modulo sums 
$M_{[c]}:=(\oplus_{i=1}^c M_{i,1}, \ldots, \oplus_{i=1}^c M_{i,k})$ as Fig. \ref{FT},
where $\oplus_{i=1}^c M_{i,1}:=M_{1,j}\oplus \cdots \oplus M_{c,j}$
and $[c]:=\{1, \ldots, c\}$.
Here, to keep the secrecy, it is required that Receiver cannot obtain any information except for the modulo sums.
In this way, we can compute the modulo sums.
As a simple case \cite{GBP1,GBP2}, we may consider the channel where
the output signal $Y$ of the channel is given as
\begin{align}
Y= X_1\oplus \cdots \oplus X_c\oplus N \Label{MYU}
\end{align}
with the $i$-th input variables $X_i$ and the noise variable $N$, where $\oplus$ is the modulo sum. 

\begin{figure}
\begin{center}
\includegraphics[scale=0.35]{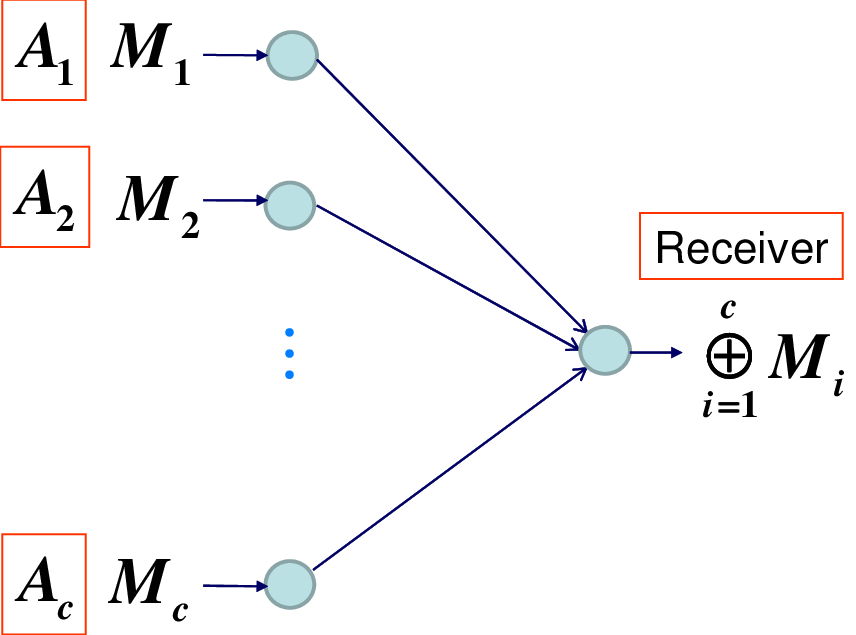}
\end{center}
\caption{Secure computation of modulo sum via multiple access channel.
}
\Label{FT}
\end{figure}%

When all senders encode their message by using the same linear code,
the receiver cannot obtain any information with respect to each sender's message
because each message is subject to the uniform distribution.
Hence, the code is decodable under the noise $N$,
the receiver can recover the modulo sums $(\oplus_{i=1}^c M_{i,1}, \ldots, \oplus_{i=1}^c M_{i,k})$.
However, the real channel does not has such a form.
A realistic example is a Gaussian MAC, whose typical case is given as follows.
The receiving signal $Y$ takes values in the set of real numbers $\RR$, and 
the $i$-th sender's signal $X_i$ takes values in discrete values $\{0,1,\ldots, p-1\}$.
Then, a typical channel is given as
\begin{align}
Y=E (X_1+\cdots +X_c)+N \in \RR,\Label{YFR}
\end{align}
where $E$ is a constant and $N$ is the Gaussian variable with average 0 and variance $v$.
In this case, a part of information for the message of each sender might be leaked to the receiver. 
To cover Gaussian MAC, this paper addresses
a general MAC.

If the secrecy condition is not imposed and the channel is a Gaussian MAC, 
this problem is a simple example of computation-and-forward \cite{Nazer08,Nazer11,Nazer16,Niesen12,Feng13,Tunali12,Narayanan15,Vazquez14}.
Due to the secrecy condition, 
we cannot realize this task by their simple application.
In the case of $c=2$,
since the modulo sums give the information of the other player,
the protocol enables the two players to exchange their messages without information leakage to Receiver \cite{Ren,He1,He2,Vatedka,Zewail,KWV}
when Receiver broadcasts the modulo sum to both players.
Application of such a protocol is discussed in several network models \cite{H-ITW}.
Our protocol can be regarded as an multi-player extension with a generic setting
of the protocol given in \cite{Ren,He1,He2,Vatedka,Zewail,KWV}.

Further, when we can realize secure transmission of modulo sum with $k$ transmitter,
multiplying $-1$ to the receiver's variable,
we can realize secure modulo zero-sum randomness among $k+1$ players,
that is the $k+1$ random variables with zero-sum condition and certain security conditions \cite[Section 2]{SMZS}\cite[Section II]{SMZS2}.
As is explained in \cite{SMZS,SMZS2},
secure modulo zero-sum randomness
can be regarded as a cryptographic resource for several cryptographic tasks including
secure multi-party computation of homomorphic functions,
secret sharing without secure communication channel,
and a cryptographic task, multi-party anonymous authentication.
In this sense, secure transmission of modulo sum is 
needed in the viewpoint of cryptography.

In this paper, 
we formulate the secure modulo sum capacity 
as the maximum of the secure transmission of modulo sum via multiple access channel.
When the multiple access channel satisfies a symmetric condition,
the secure modulo sum capacity equals the capacity of a certain single access channel.
In a general setting,
we derive the lower bound of the secure modulo sum capacity. 
Further, we give several realistic examples by using Gaussian MAC, which can be realized in wireless communication.

The remaining part of this paper is organized as follows.
Section II prepares notations used in this paper.
Section III defines the secure modulo sum capacity, and 
gives two theorems, which discusses the symmetric case and the general case.  
Section IV gives 
the proof for the semi-symmetric case.  
Section V discusses the Gaussian MAC.

\section{Notations}
In this section, we prepare notations and information quantities used in this paper.
Given a joint distribution channel $P_{Y,Z_1, Z_2}$ over the product system of 
a finite discrete set ${\cal Z}_1\times {\cal Z}_2$ and a continuous set ${\cal Y}$,
we denote the conditional probability density function of 
$P_{Y|Z_1,Z_2}$ by $p_{Y|Z_1,Z_2}(y|z_1,z_2) $.
Then, we define the conditional distribution $P_{Y| Z_2}$ 
over a continuous set ${\cal Y}$ conditioned in the discrete set ${\cal Z}_2$
by the conditional probability density function
$p_{Y|Z_2}(y|z_2):= \sum_{z_2\in {\cal Z}_2} P_{Z_2}(z_2) p_{Y|Z_1,Z_2}(y|z_1,z_2)$.
Then, we define two types of the R\'{e}nyi conditional mutual information
$I_{1+s}( Y;Z_1|Z_2)$
\begin{align}
&s I_{1+s}( Y;Z_1|Z_2)\nonumber \\
:=&
\log 
\sum_{z_1,z_2} 
P_{Z_1,Z_2}(z_1,z_2)
\int_{{\cal Y}}  
\frac{p_{Y|Z_1,Z_2}(y|z_1,z_2)^{1+s}} 
{p_{Y|Z_2}(y|z_2)^{s}} dy \\
& \frac{s}{1+s} I_{1+s}^{\downarrow}( Y;Z_1|Z_2)\nonumber \\
:=&
\log 
\sum_{z_2}
P_{Z_2}(z_2)
\int_{{\cal Y}}  
\Big(
\sum_{z_1}
P_{Z_1|Z_2}(z_1|z_2)
p_{Y|Z_1,Z_2}(y|z_1,z_2)^{1+s}
\Big)^{\frac{1}{1+s}}
dy 
\end{align}
for $s>0$.
Since $\lim_{s \to 0}s I_{1+s}( Y;Z_1|Z_2)=
\lim_{s \to 0}s I_{1+s}^{\downarrow}( Y;Z_1|Z_2)=0$,
taking the limit $s \to 0$, we have
\begin{align}
\lim_{s \to 0}\frac{s I_{1+s}( Y;Z_1|Z_2)}{s} =
\lim_{s \to 0}\frac{s I_{1+s}^{\downarrow}( Y;Z_1|Z_2)}{s} =
I( Y;Z_1|Z_2),
\end{align}
where 
$I( Y;Z_1|Z_2)$ expresses the conditional mutual information. 
Also, we have
\begin{align}
 I_{1+s}^{\downarrow}( Y;Z_1|Z_2)\le 
I_{1+s}( Y;Z_1|Z_2) \le  I_{\frac{1}{1-s}}^{\downarrow}( Y;Z_1|Z_2).\Label{HY1}
\end{align}
The concavity of the function $x \mapsto x^{\frac{1}{1+s}}$ yields
\begin{align}
e^{\frac{s}{1+s} I_{1+s}^{\downarrow}( Y;Z_1|Z_2,Z_3)}
\le e^{\frac{s}{1+s} I_{1+s}^{\downarrow}( Y;Z_1,Z_2|Z_3)}.\Label{HY2}
\end{align}

Given a channel $P_{Y|Z_1, Z_2,Z_3}$ from the finite discrete set ${\cal Z}_1\times {\cal Z}_2\times {\cal Z}_3$
to a continuous set ${\cal Y}$,
when the random variables $Z_1,Z_2,Z_3$ are generated subject to the uniform distributions,
we have a joint distribution among $Y,Z_1,Z_2$.
In this case, we denote the Renyi conditional mutual information and the conditional mutual information
by $I_{1+s}( Y;Z_1|Z_2)[P_{Y|Z_1,Z_2,Z_3}]$ and $I( Y;Z_1|Z_2)[P_{Y|Z_1,Z_2,Z_3}]$,
respectively.

Further, 
we denote the vector $(x_i)_{i \in {\cal I}}$ by $\vec{x}_{{\cal I}}$.
When each element $x_i$ belongs to the same vector space, 
we use the notation
\begin{align}
x_{{\cal I}}:=\oplus_{i \in {\cal I}} x_i.
\end{align}

\section{Secure Modulo Sum Capacity}
Consider $c$ players and Receiver.
Player $A_i$ has a random variable $M_i \in \FF_q^k$ and sends it to Receiver,
where $M_i$ is independently subject to the uniform distribution.
The task of Receiver is computing the modulo sum $M_{[c]} \in \FF_q^k$.
Also, it is required that Receiver obtains no other partial information of $\{M_i\}_{i=1}^c$.
That is, when Receiver's information is denoted by $\bY \in {\cal Y}^n$,
we impose the following condition 
\begin{align}
I( \bY; \vec{M}_{{\cal J}} ) \cong 0
\end{align}
for any non-empty set ${\cal J} \subsetneqq [c]$,
where $\vec{M}_{{\cal J}}:=(M_i)_{i \in {\cal J}}$.
We call this task secure transmission of common reference string.

To realize this task, we employ a multiple access channel (MAC) $W$
with $c$ input alphabets ${\cal X}_1, \ldots, {\cal X}_c$ 
and an output alphabet ${\cal Y}$, which might be a continuous set.
Here, the $i$-th input alphabet ${\cal X}_i$ is under the control of Player $A_i$.
That is, given an element $ (x_1, \ldots, x_c)\in 
{\cal X}_1\times \ldots\times {\cal X}_c$,
we have the output distribution $W_{(x_1, \ldots, x_c)}$ on ${\cal Y}$.
When Receiver is required to compute the modulo sum of elements of $\FF_q^{k_n}$ 
with the above security condition,
we employ $n$ uses of the multiple access channel $W$.
Here, we use the bold font like $\bY,\bx_j,\bX_j,\vec{\bY}_{{\cal J}}$ 
to express $n$ symbols while
the italic font like $Y,x_j,X_j,\vec{Y}_{{\cal J}}$ expresses single symbol.
Receiver's variable $\bY \in {\cal Y}^n$ obeys the output distribution of the channel $W^{(n)}$, 
which is defined as the distributions 
$W_{(\bx_1, \ldots, \bx_c)}^{(n)}:=W_{(x_{1,1}, \ldots, x_{c,1})}
\times \cdots \times W_{(x_{1,n}, \ldots, x_{c,n})}$ 
on ${\cal Y}^{n}$
for $ (\bx_1, \ldots, \bx_c)\in {\cal X}_1^n\times \ldots\times {\cal X}_c^n$.

The encoder is given as a set of stochastic maps $\Phi_{i,e}^{(n)} $ from $\FF_q^{k_n}$ to ${\cal X}_i^n$ with $i=1, \ldots,c$.
The decoder is given as a map $\Phi_{d}^{(n)}  $ from  ${\cal Y}^n $
to $\FF_q^{k_n}$.
The tuple $((\Phi_{i,e}^{(n)})_{i=1}^c,\Phi_{d}^{(n)})
$ is simplified to $\Phi^{(n)}$ and is called a code.
The performance of code $\Phi^{(n)}$ is given by the following two values.
One is the decoding error probability
\begin{align}
\epsilon(\Phi^{(n)}) 
:=\frac{1}{q^{c k_n}}
\sum_{
\substack{
m_1,\ldots,m_c,m'\\ 
m' \neq  \oplus_{j=1}^c m_j}} 
\mathbb{E}_{(\Phi_{i,e}^{(n)})_{i=1,\ldots, c}}
\Big[W^{(n)}_{(\Phi_{i,e}^{(n)} (m_i))_{i=1,\ldots, c}}
( (\Phi_{d}^{(n)})^{-1}(m') ) \Big],
\end{align}
where the sum is taken with respect to 
$m_1,\ldots,m_c,m' \in \FF_q^{k_n} $ 
with the condition $m' \neq  \oplus_{j=1}^c m_j$, and
$\mathbb{E}_{(\Phi_{i,e}^{(n)})_{i=1,\ldots, c}}$ expresses the expectation 
with respect to the random variables $(\Phi_{i,e}^{(n)})_{i=1,\ldots, c}$.
The other is the leaked information
\begin{align}
I_{{\cal J}} ( \Phi^{(n)}):=
I(\bY; \vec{M_{\cal J}} )
\end{align}
for a non-empty set ${\cal J} \subsetneqq [ c ]$. 
The transmission rate is given as $ \frac{k_n\log q}{n}$.

When a sequence of codes $\{\Phi^{(n)}\}$ satisfies the conditions
\begin{align}
\epsilon(\Phi^{(n)}) \to 0
\end{align}
and 
\begin{align}
I_{\cal J}( \Phi^{(n)}) \to 0 
\end{align}
for a non-empty set ${\cal J} \subsetneqq [ c ]$,
the limit 
$ \limsup_{n\to \infty}  \frac{k_n\log q}{n}$
is called an achievable rate.
The {\it secure modulo sum capacity} $C_{\SMS}(W)$ 
is defined as the supremum of achievable rates
with respect to the choice of the sequence codes as well as the prime power $q$.

Now, we assume that ${\cal X}_i$ is an $l$-dimensional vector space over a finite field $\FF_q$.
Also, ${\cal Y}$ is a symmetric space for $\FF_q$, i.e., 
for an element $x \in \FF_q^l$, there is a function $f_x$ on ${\cal Y}$ 
such that $ f_{x_1} \circ f_{x_2}=f_{x_1 \oplus x_2}$.
Here, the collection of functions $f=(f_x)$ is called an action of $\FF_q^l$ 
on ${\cal Y}$.  
A multiple access channel $W$ is called {\it symmetric}
when the relation 
\begin{align}
W_{x_1,\ldots, x_c }(B_{f,x})
=W_{x_1\oplus x,\ldots, x_c }(B)
=\cdots
=W_{x_1,\ldots, x_c\oplus x }(B)
\end{align}
holds for a measurable subset $B \subset {\cal Y}$, 
where $B_{f,x}:= \{ f_x(y) \}_{y \in B} $.

For a symmetric multiple access channel $W$, 
we have $W_{x_1,\ldots,x_c}=W_{x_1+\ldots+x_c,0,\ldots,0}$, and
define 
the symmetric single access channel $W_S$ as
\begin{align}
W_{S,x}:=W_{x,0,\ldots,0},
\end{align}
which includes the channel \eqref{MYU}.

\begin{theorem}
For a symmetric multiple access channel $W$, we have
\begin{align}
C_{\SMS}(W)=C(W_S),
\end{align}
where $C(W_S)$ is the channel capacity of the channel $W_S$.
\end{theorem}

Its proof is very simple.
In this case, the transmission of modulo sum is the same as 
the information transmission under 
the symmetric single access channel $W_S$
when the senders employ a common algebraic encoder.
The symmetric condition guarantees the secrecy condition
when the senders employ a common algebraic encoder
because 
the two distributions $W_{x_1,\ldots, x_c }$ and $W_{x_1',\ldots, x_c' }$ cannot be distinguished
when $x_1 \oplus\ldots\oplus x_c=x_1' \oplus\ldots\oplus x_c' $.
Hence, the rate $C(W_S)$ is achieved.
The converse part can be shown as follows.
Consider the case when the messages by Senders $A_{2}, \ldots, A_c$ are fixed.
In this case, the problem is reduced to the message transmission from Sender $A_1$ to Receiver.
Due to the above symmetric condition, the channel with this case from Sender $A_1$ to Receiver
is also $W_S$. Hence, it is impossible to exceed the rate $C(W_S)$.
Hence, the converse part follows.

However, it is not so easy to realize a symmetric multiple access channel $W$.
Here, we remark the relation with the existing papers \cite{GBP1,GBP2}, which address a similar channel.
However, they consider the information leakage to the third party with respect to the modulo sum.
This paper considers the information leakage to the receiver with respect to the message of each player etc.
To address this problem, we 
denote $I(Y;X_{[c]})[W]$ and $ I(Y;X_i|\vec{X}_{{\cal J}_i})[W]$
by the mutual information 
$I(Y;X_{[c]})$ and $ I(Y;X_i|\vec{X}_{{\cal J}_i})$
when $X_1, \ldots, X_c$ is subject to the uniform distribution
and $Y$ is their output of the multiple access channel $W$.
We define $ I(Y; (X_j\ominus X_{j'})_{j,j'\in {\cal J}^c  }|\vec{X}_{{\cal J}_i  })[W ]$
and $I(Y; (X_j\ominus X_{j'})_{j,j'\in {\cal J}^c  }|\vec{X}_{{\cal J}_i  })[W ]$
in the same way, where $\ominus$ expresses the minus in the sense of the finite field $\FF_q$.
We prepare notations $ I(W):=I(Y;X_{[c]})[W]$,
${\cal J}_j:= \{ i \in {\cal J}| i\neq j\}$,
and $\vec{X}_{{\cal J}}:= (X_i)_{i \in {\cal J}}$.

For the general case, instead of Theorem 1, we have the following theorem.
\begin{theorem}\Label{T-2}
Assume that $ {\cal X}_i$ is an $l$-dimensional vector space over $\FF_q$ as ${\cal X}$.
When a multiple access channel $W$ satisfies
\begin{align}
 I(W) &>   I(Y;X_i|\vec{X}_{{\cal J}_i})[W] \Label{IN1}\\
 I(W)(1+|{\cal J}^c|) &>  I(Y;X_i \vec{X}_{{\cal J}^c  }|\vec{X}_{{\cal J}_i  })[W ]\Label{IN2}\\
 I(W)|{\cal J}^c| &>  I(Y; (X_j\ominus X_{j'})_{j,j'\in {\cal J}^c  }|\vec{X}_{{\cal J}_i  })[W ]\Label{IN3}
\end{align}
for ${\cal J} \subsetneqq [c]$ and $i \in {\cal J}$,
we have
\begin{align}
C_{\SMS}(W)
\ge  
\min_{i\in {\cal J} \subsetneqq [c]} 
\min \Big(&
I(W) -  I(Y;X_i|\vec{X}_{{\cal J}_i})[W],
\nonumber \\
& (|{\cal J}^c| +1)I(W) -I (Y; X_i ,\vec{X}_{{\cal J}^c} |\vec{X}_{{\cal J}_i} )[W] ,
\nonumber \\
& |{\cal J}^c| I(W) -I (Y; X_i , 
(X_j\ominus X_{j'})_{j,j' \in {\cal J}^c}| \vec{X}_{{\cal J}_i} ) [W]
\Big).
\end{align}

Under the notation $I(Y;X_i|\vec{X}_{{\cal J}_i})[W]$,
the symbols $W$, $Y$, $X_i$, $\vec{X}_{{\cal J}_i}$, and  
$\vec{X}_{[c] \setminus{\cal J}}$ correspond to 
$P_{Y|Z_1,Z_2,Z_3}$, $Y$, $Z_1$, $Z_2$, and $Z_3$
in the definition of $I( Y;Z_1|Z_2)[P_{Y|Z_1,Z_2,Z_3}]$,
respectively.
\end{theorem}
Due to Lemma 2 of Appendix, the maximum 
$\max_{j\in {\cal J} \subsetneqq [c]}  I(Y;X_j|\vec{X}_{{\cal J}_j})[W]$
is realized when $| [c]\setminus {\cal J}|=1$.
The mutual information $I(W)$ is a generalization of the achievable rate with $\FF_2$ given in \cite{Ullah}.

To hide message $M_i$, other messages work as scramble variables for $M_i$. 
If the variables of other senders are subject to the uniform distribution on ${\cal X}^n$, 
Condition \eqref{IN1} is sufficient for the secrecy. 
However, the variables of other senders are not subject to the above uniform distribution on ${\cal X}^n$ in our code. 
Moreover, when we fix the message $M_i$ and we take the average for other messages, 
the output cannot be regarded as the output under the above uniform distribution in general because the message size is not sufficiently large. 
Hence, we need to care the deviation of the inputs of other senders. 
Conditions \eqref{IN2} and \eqref{IN3} are the additional conditions to cover the secrecy even with such deviation.

\section{Proof of Theorem \ref{T-2}}\Label{S4}
\noindent{\bf Step (1):}\quad 
We construct our code randomly as follows.
As the first step, we choose a sequence of integers $k_n$ such that
the minimum among the differences between LHS and RHS in inequalities 
 \eqref{IN1} -- \eqref{IN3} equals $\epsilon+\lim_{n \to \infty} \frac{k_n}{n} \log q $,
 where $\epsilon>0$ is sufficiently small real number.
Then, we choose a sequence of integers $k_n'$ such that
$I(W)-\frac{\epsilon}{2} =
\lim_{n \to \infty}\frac{k_n+k_n'}{n}\log q$.
Then, we have the following inequalities;
\begin{align}
\lim_{n \to \infty}\frac{k_n+k_n'}{n}\log q &<I(W) \Label{8-20-4}\\
\lim_{n \to \infty} \frac{k_n'+ (k_n+k_n')|{\cal J}^c|}{n} \log q
&> I (Y; X_i ,\vec{X}_{{\cal J}^c} |\vec{X}_{{\cal J}_i} )[W]  \Label{8-20-5}\\
\lim_{n \to \infty} \frac{k_n'}{n} \log q
&> I (Y; X_i | \vec{X}_{{\cal J}_i} ) [W] \Label{8-20-6},
\end{align}
and
\begin{align}
\lim_{n \to \infty} \frac{k_n'+ (k_n+k_n')(|{\cal J}^c|-1)}{n} \log q 
>  I (Y; X_i , (X_j\ominus X_{j'})_{j,j' \in {\cal J}^c}| \vec{X}_{{\cal J}_i} ) [W]
 \Label{8-20-7}
\end{align}
for ${\cal J} \subsetneqq [c]$ and $i \in {\cal J}$.
Let $L_i$ be
a scramble random variable with the dimension $k_n'$ for $i=1, \ldots, c$,
which is assumed to be subject to the uniform distribution.

Then, 
we randomly choose a pair of an invertible linear map $G_1$ 
from $ \FF_q^{k_n+k_n'} \times \FF_q^{nl -(k_n+k_n')}$ to ${\cal X}^n $
and a linear map $G_2$ from ${\cal X}^n $ to $ \FF_q^{nl -(k_n+k_n')} $
such that $G_2 \circ G_1 (F, E)=E$ and the universal2 condition
\begin{align}
{\rm Pr} \{ G_2 (x)=G_2 (x') \} &\le \frac{1}{q^{nl -(k_n+k_n')}} \Label{HY8}
\end{align}
holds for $x\neq x' \in {\cal X}^n$.
Also, we randomly choose a pair of an invertible linear map $G_3$ 
from $ \FF_q^{k_n} \times \FF_q^{k_n'}$ to $\FF_q^{k_n+k_n'}$
and a linear map $G_4$ from $\FF_q^{k_n+k_n'}$ to $ \FF_q^{k_n}$
such that $G_4 \circ G_3 (M, L)=M$ and
the universal2 condition
\begin{align}
{\rm Pr} \{ G_4 (f)=G_4 (f') \} &\le \frac{1}{q^{k_n'}}\Label{HY9}
\end{align}
holds for $f\neq f' \in \FF_q^{k_n+k_n'}$.
Hence, $G_1$ and $G_3$ uniquely determine $G_2$ and $G_4$, respectively.

We also randomly and independently choose the random variable $E_i$ subject to the uniform distribution on $\FF_q^{nl -(k_n+k_n')}$ for $i=1, \ldots, c$.
Hence, $G_1$, $G_3$, and $E_1, \ldots, E_c$ are priorly shared among players and Receiver.
Then, we consider the following protocol.
That is, using the random variable $G_1$, $G_3$, and $E_i$,
we define the encoder $ \Phi_{i,e}^{(n)}$ as follows.
For a given $M_i$, Player $A_i$ randomly chooses the scramble random variable $L_i$.
That is, each player $A_i$ transmits $G_1(G_3(L_i,M_i),E_i)$.
In this case, since $G_1$, $G_3$, and $E_i$ are shared among Senders and Receiver, 
the encoder of Sender $A_i$ is given as an affine map
$(L_i,M_i) \mapsto G_1(G_3(L_i,M_i),E_i)$.
The detail analysis for affine maps are explained in Appendix A of \cite{HaVa}. 
Here, the variables $G_1,G_3, \vec{M},\vec{L},\vec{E}$ are assumed to be independent of each other.
Then, we have the relation
\begin{align}
F_i= G_3(L_i,M_i), \quad 
\bX_i= G_1(F_i,E_i).
\end{align}
Since $\vec{E}$ is subject to the uniform distribution on 
$ \FF_q^{c(nl -(k_n+k_n'))}$, 
even when the maps $G_1$ and $G_3$ are fixed,
$\vec{\bX}$ is subject to the uniform distribution on 
$ {\cal X}^{cn}$.
Receiver receives the random variable $\bY \in {\cal Y}^n$ that  depends only on $\vec{\bX}$.
That is, we have the Markov chain $(G_1,G_3, \vec{M},\vec{L},\vec{E})- \vec{\bX} -\bY$.
Due to \eqref{8-20-4}, Receiver can decode 
$\oplus_{j=1}^c F_j=G_3(L_{[c]},M_{[c]})$, i.e., $M_{[c]}$ and $L_{[c]}$
from $\bY$ by using the knowledge $E_{[c]}=\oplus_{i=1}^c E_i$ for the coset.
The proof is similar to the same way as Appendix B of \cite{HaVa} and is different from \cite{Lim}.
The detail is explained in Appendix \ref{AC}.

\noindent{\bf Step (2):}\quad 
We divide the leaked information to several parts that can be bounded.
Given ${\cal J} \subset [ c]$, we discuss the leaked information
for $M_{\cal J}:= \sum_{i \in {\cal J}} M_i $.
For $\vec{M}_{\cal J}:= (M_i)_{i \in {\cal J}}$, using 
$\tilde{{\cal J}}_i:= \{ j \in {\cal J}| j<i\}$,
we have
\begin{align}
&I( \bY;\vec{M}_{\cal J} |G_1,G_3,\vec{E}) \nonumber\\
=&\sum_{i \in {\cal J}}
I(\bY;M_i | \vec{M}_{\tilde{{\cal J}}_i} ,G_1,G_3,\vec{E})
\nonumber\\
\stackrel{(a)}{\le}& 
\sum_{i \in {\cal J}}
I(\bY;M_i | \vec{M}_{{\cal J}_i} ,\vec{L}_{{\cal J}_i} ,G_1,G_3,\vec{E}) \nonumber\\
=&\sum_{i \in {\cal J}}
I(\bY;M_i  | \vec{M}_{{\cal J}_i} ,\vec{L}_{{\cal J}_i},G_1,G_3,\vec{E}_{{\cal J}_i},
\vec{E}_{{\cal J}_i^c} ) \nonumber\\
=& \sum_{i \in {\cal J}}
I(\bY;M_i | \vec{\bX}_{{\cal J}_i} ,G_1,G_3,\vec{E}_{{\cal J}_i^c} )\nonumber \\
=&\sum_{i \in {\cal J}}
I(\bY; M_i  | \vec{E}_{{\cal J}^c} ,\vec{\bX}_{{\cal J}_i} ,G_1,G_3,E_i) ,
\Label{8-20-3}
\end{align}
where $(a)$ follows from Eq. \eqref{E2} in Appendix.

\noindent{\bf Step (3):}\quad 
We focus on the randomness of the choice of $G_3$. 
Then, for $s \in [0,\frac{1}{2}]$, we have
\begin{align}
& I(\bY; M_i | G_3, G_1, E_i, \vec{E}_{{\cal J}^c} ,\vec{\bX}_{{\cal J}_i} ) 
\nonumber\\
\stackrel{(a)}{\le} & 
 q^{-s k_n'}
e^{s I_{1+s}(\bY; F_i | G_1, E_i, \vec{E}_{{\cal J}^c} ,\vec{\bX}_{{\cal J}_i} ) } 
\nonumber\\
\stackrel{(b)}{\le} & 
 q^{-s k_n'}
e^{s I_{\frac{1}{1-s}}^{\downarrow}
(\bY; F_i | G_1, E_i, \vec{E}_{{\cal J}^c} ,\vec{\bX}_{{\cal J}_i} ) } 
\nonumber\\
\stackrel{(c)}{\le} & 
 q^{-s k_n'}
e^{s I_{\frac{1}{1-s}}^{\downarrow}
(\bY; F_i, E_i | G_1, \vec{E}_{{\cal J}^c} ,\vec{\bX}_{{\cal J}_i} ) } 
\nonumber\\
\stackrel{(d)}{=} &  
q^{-s k_n'}
e^{s I_{\frac{1}{1-s}}^{\downarrow}
(\bY; \bX_i | G_1, \vec{E}_{{\cal J}^c} ,\vec{\bX}_{{\cal J}_i} ) } 
\nonumber\\
\stackrel{(e)}{\le} & 
q^{-s k_n'-s (k_n+k_n')|{\cal J}^c|}
e^{s I_{\frac{1}{1-s}}^{\downarrow}(\bY; \bX_i ,\vec{\bX}_{{\cal J}^c} |\vec{\bX}_{{\cal J}_i} ) } 
+q^{-s k_n'} e^{s I_{\frac{1}{1-s}}^{\downarrow}
(\bY; \bX_i | \vec{\bX}_{{\cal J}_i} ) } 
\nonumber\\
&+q^{-s k_n'-s (k_n+k_n')(|{\cal J}^c|-1)}
 e^{s I_{\frac{1}{1-s}}^{\downarrow}
(\bY; \bX_i , (\bX_j\ominus\bX_{j'})_{j,j' \in {\cal J}^c}| \vec{\bX}_{{\cal J}_i} ) } ,
\Label{HY3}
\end{align}
where $(a)$ follows from Theorem 4 of \cite{Hayashi2011} and \eqref{HY9};
$(b)$ follows from the second inequality in \eqref{HY1};
$(c)$ follows from \eqref{HY2};
$(d)$ follows from the fact that the pair $(E_i,F_i)$ and $\bX_i$
uniquely determine each other;
and
$(e)$ follows from Appendix B.
Due to Conditions \eqref{8-20-5} -- \eqref{8-20-7},
all the terms in \eqref{HY3} go to zero exponentially.
Hence, we obtain Theorem \ref{T-2}.
\endproof

\begin{remark}
Theorem \ref{T-2} cannot be shown by simple application of the result of wire-tap channel \cite{Wyner} as follows.
Consider the secrecy of the message $M_1$ of player $A_1$.
In this case, if other players transmit elements of ${\cal X}^n$ with equal probability,
the channel from player $A_1$ to Receiver 
is given as $n$-fold extension of the channel 
$x \mapsto \frac{1}{q}\sum_{x_2, \ldots, x_c}W_{x,x_2, \ldots, x_c}$, which enables us to directly apply the result of wire-tap channel.
However, other players transmit elements of the image of  $G$, which is a subset of ${\cal X}^n$, with equal probability.
Hence, the channel from player $A_1$ to Receiver does not have the above simple form.
Therefore, we need more careful discussion.
\end{remark}

\begin{remark}
In this proof, Receiver does not use the knowledge $E_1, \ldots, E_c$ except for 
$E_{[c]}:=\oplus_{i=1}^c E_i$.
Hence, there is a possibility that the transmission rate can be improved by using this knowledge.
\end{remark}

\section{Gaussian MAC}

\subsection{Real case}
First, we consider the Gaussian MAC \eqref{YFR}, in which
${\cal Y}$ is $\RR$ and 
${\cal X}_i$ is $\FF_p$.
Using the Gaussian distribution $P_{a}$ with the average $a$ and variance $v$ on $\RR$,
we have the channel as $W_{x_1,\ldots, x_c }:=P_{E(x_1+\ldots+ x_c)} $.
In this case, we have
\begin{align}
&I(W)= 
h( \frac{\sum_{x_1, \ldots,x_c} P_{(x_1+ \ldots+x_c)E}}{p^c})\nonumber \\
&\hspace{8ex} -\sum_{j=0}^p \frac{1}{p} 
h(\sum_{i=1}^c  \frac{\alpha(i,j)}{p^{c-1}} P_{(j+ip)E}   )\\
&\max_{j\in {\cal J} \subsetneqq [c]}  I(Y;X_j|\vec{X}_{{\cal J}_j})[W]
\stackrel{(a)}{=}
 I(Y;X_1|\vec{X}_{\{2, \ldots, c-1\}})[W] \Label{cc} \\
&= h( \frac{\sum_{j,j'=0}^{p-1} P_{jj'E}}{p^2})-h( \frac{\sum_{j=0}^{p-1} P_{jE}}{p}),
\end{align}
where $h$ is the differential entropy
and $\alpha(i,j)$ is the number of elements $(x_1, \ldots,x_c)\in \{0,\ldots,p-1\}^c $ satisfying 
$\sum_{i'=1}^c x_{i'}=j+ip$.
Here, $(a)$ follows the discussion after Theorem 2.
When $p=2$ and $c=3$, our lower bound is
$
h(\frac{P_0+2P_E+2P_{2E}+P_{3E}}{6})
-h(\frac{P_0+2 P_{2E}}{3})
-h(\frac{P_0+2P_E+P_{2E}}{4})
+h(\frac{P_0+P_E}{2})$, which
is numerically calculated as Fig. \ref{F1}.

When $E$ goes to infinity, the distributions $\{P_{iE}\}_{i=0}^{c(p-1)}$ can be distinguished so that
$I(W)$ goes to $\log p$
and $ h( \frac{\sum_{j,j'=0}^{p-1} P_{jj'E}}{p^2})-H( \frac{\sum_{j=0}^{p-1} P_{jE}}{p})$ 
goes to $H(Q_p) -\log p= \sum_{j=-p+1}^{p-1}\frac{|j-p|}{p^2}\log \frac{p}{|j-p|}$,
where the distribution $Q_p$ on $\{-p+1, 2,\ldots, p-1\}$ is defined as $Q_p(j):= \frac{|j-p|}{p^2}$.
Hence, our lower bound of the secure modulo sum capacity goes to 
$\sum_{j=-p+1}^{p-1}\frac{|j-p|}{p^2}\log |j-p|$.
For example, when $ p=2$, it is $\frac{1}{2}\log 2$.

\begin{figure}[t]
\begin{center}
\includegraphics[scale=0.9]{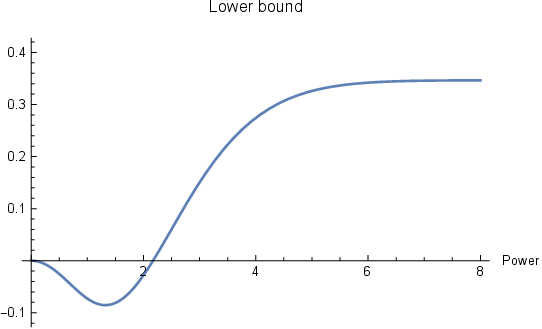}
\end{center}
\caption{Our lower bound of secure modulo sum capacity with $p=2$
when 
three senders send information ($c=3$) and the variance $v$ is $1$.
The base of logarithm is chosen to be $e$.
The horizontal axis expresses the power $E$. 
The vertical axis expresses our lower bond
$h(\frac{P_0+2P_E+2P_{2E}+P_{3E}}{6})
+h(\frac{P_0+P_E}{2})
-h(\frac{P_0+2 P_{2E}}{3})
-h(\frac{P_0+2P_E+P_{2E}}{4})$
This lower bound  approaches $\frac{1}{2}\log 2$ and is positive with $E \ge 2.1715$.
}
\Label{F1}
\end{figure}%

\subsection{Complex case}
Next, we consider the case when ${\cal Y}$ is $\CC$,
and discuss the case with ${\cal X}_i=\FF_7$ as a typical case \cite{Tunali12,Vazquez14}.
Let $P_{a}$ be the Gaussian distribution with the average $a$ and variance $v$ on $\CC$.
We define the map $u$ from $\FF_7$ to $\CC$ as follows.
\begin{align}
u(0)&=0,~ 
u(6)=E,~
u(7)= Ee^{\frac{2\pi }{3} i },~
u(8)= Ee^{\frac{\pi }{3} i },\nonumber\\
u(1)&=Ee^{\frac{4\pi }{3} i },~
u(3)= Ee^{\frac{5\pi }{3} i },~
u(2)= -E.
\end{align}
Then, we define 
$W_{x_1,\ldots, x_c }:=P_{(u(x_1)+\ldots+ u(x_c))E} $.
Hence, we have \eqref{cc}.
Due to the same reason as the above case, 
when $E$ goes to infinity, $I(W)$ goes to $\log 7$
and $ I(Y;X_1|\vec{X}_{\{2, \ldots, c-1\}})[W] $
goes to $\sum_{j=-6}^{6}\frac{|j-7|}{7^2}\log \frac{7}{|j-7|}$,
Hence, our lower bound of the secure modulo sum capacity goes to 
$\sum_{j=-6}^{6}\frac{|j-7|}{7^2}\log |j-7|$.

\section{Conclusion}
To discuss the generation of common reference string via a multiple access channel, we have introduced the secure modulo sum capacity for a multiple access channel.
We have shown that the secure modulo sum capacity equals the channel capacity 
when the multiple access channel satisfies the symmetric condition.
Since the symmetric condition does not hold in a natural setting, we have derived a lower bound for 
the secure modulo sum capacity in a general setting.
We have examined this lower bound under the real and complex Gaussian MAC
with numerical analysis.

\section*{Acknowledgments}
The author is very grateful to Professor \'{A}ngeles V\'{a}zquez-Castro and Professor Takeshi Koshiba
for their helpful discussions.

\appendices
\section{Conditional mutual information}
We prepare a lemma for conditional mutual information.
\begin{lemma}
\begin{align}
I(A;B|CD)&=
-I(B;C|D)
+I(A;B|D)
+I(B;C|AD) .\Label{E1}
\end{align}
In particular, when $I(B;C|D)=0$, we have
\begin{align}
I(A;B|D)\le I(A;B|CD).\Label{E2}
\end{align}
\end{lemma}
\begin{proof}
This equation can be shown as follows.
\begin{align}
&I(A;B|CD)+I(B;C|D)
=
I(AC;B|D) \nonumber\\
=&I(A;B|D)
+I(B;C|AD).
\end{align}
\end{proof}

\section{Proof of Step $(e)$ in \eqref{HY3}}
Given $\vec{\bx}_{{\cal J}^c} \in {\cal X}^{|{\cal J}^c|}$ and $G_1$,
we introduce the following conditions for 
$\vec{\bx}_{{\cal J}^c}' \in {\cal X}^{|{\cal J}^c|}$.
\begin{description}
\item[(C1)]
$\vec{\bx}_{{\cal J}^c}'\neq \vec{\bx}_{{\cal J}^c}$.
\item[(C2)]
$(G_2(\bx_j'))_{j\in {\cal J}^c}=(G_2(\bx_j))_{j\in {\cal J}^c}$. 
\item[(C3)]
$(\bx_j' \ominus\bx_{j'}')_{j,j'\in {\cal J}^c} \neq (\bx_j \ominus\bx_{j'})_{j,j'\in {\cal J}^c}$.  
\item[(C4)]
$(\bx_j'\ominus\bx_{j'}')_{j,j'\in {\cal J}^c} = (\bx_j\ominus\bx_{j'})_{j,j'\in {\cal J}^c}$.  
\end{description}
In the following, we denote the sum with respect to 
$\vec{\bx}_{{\cal J}^c}' \in {\cal X}^{|{\cal J}^c|}$ 
under the conditions (C1) and (C2) by $\sum_{\vec{\bx}_{{\cal J}^c}':{\rm (C1)(C2)}}$, etc.
Using the condition \eqref{HY8}, we have
\begin{align}
\EE_{G_1}
\sum_{ \vec{\bx}_{{\cal J}^c}':{\rm (C1)(C2)(C3)} }
P_{Y|\bX_i=\bx_i , \vec{\bX}_{{\cal J}^c}=\vec{\bx}_{{\cal J}^c}' ,\vec{\bX}_{{\cal J}_i}}(y)
&\le 
\sum_{ \vec{\bx}_{{\cal J}^c}':{\rm (C1)(C3)} }
q^{(-nl+(k_n+k_n'))|{\cal J}^c|}
P_{Y|\bX_i=\bx_i , \vec{\bX}_{{\cal J}^c}=\vec{\bx}_{{\cal J}^c}' ,\vec{\bX}_{{\cal J}_i}}(y)
\nonumber \\
& \le
\sum_{ \vec{\bx}_{{\cal J}^c}' }
q^{(-nl+(k_n+k_n'))|{\cal J}^c|}
P_{Y|\bX_i=\bx_i , \vec{\bX}_{{\cal J}^c}=\vec{\bx}_{{\cal J}^c}' ,\vec{\bX}_{{\cal J}_i}}(y)
\nonumber \\
& =q^{(k_n+k_n')|{\cal J}^c|}
P_{Y|\bX_i=\bx_i ,\vec{\bX}_{{\cal J}_i}}(y) \Label{HY11}\\
\EE_{G_1}
\sum_{ \vec{\bx}_{{\cal J}^c}':{\rm (C1)(C2)(C4)} }
P_{Y|\bX_i=\bx_i , \vec{\bX}_{{\cal J}^c}=\vec{\bx}_{{\cal J}^c}' ,\vec{\bX}_{{\cal J}_i}}(y)
& \le 
\sum_{ \vec{\bx}_{{\cal J}^c}':{\rm (C1)(C4)} }
q^{-nl+(k_n+k_n')}
P_{Y|\bX_i=\bx_i , \vec{\bX}_{{\cal J}^c}=\vec{\bx}_{{\cal J}^c}' ,\vec{\bX}_{{\cal J}_i}}(y) \nonumber \\
& \le 
\sum_{ \vec{\bx}_{{\cal J}^c}':{\rm (C4) }}
q^{-nl+(k_n+k_n')}
P_{Y|\bX_i=\bx_i , \vec{\bX}_{{\cal J}^c}=\vec{\bx}_{{\cal J}^c}'' ,\vec{\bX}_{{\cal J}_i}}(y) \nonumber \\
&=
q^{(k_n+k_n')}
P_{Y|\bX_i=\bx_i , (\bX_j-\bX_{j'}=\bx_j-\bx_{j'})_{j,j' \in {\cal J}^c} ,\vec{\bX}_{{\cal J}_i}}(y).\Label{HY12}
\end{align}

To show the step $(e)$ in \eqref{HY3}, 
we define the set 
${\cal S}(G_2,\vec{E}_{{\cal J}^c}):=
\{ \vec{\bx}_{{\cal J}^c} \in {\cal X}^{|{\cal J}^c|}| G_2(\bx_j)=E_j$ for 
$ j \in {\cal J}^c\}$.
Using \eqref{HY11} and \eqref{HY12}, for $s \in [0,\frac{1}{2}]$,
we have the following relations, 
where the explanations for steps is explained later.
\begin{align}
& e^{s I_{\frac{1}{1-s}}^{\downarrow}
(\bY; \bX_i | G_1, \vec{E}_{{\cal J}^c} ,\vec{\bX}_{{\cal J}_i} ) } .
\nonumber\\
= &
\EE_{\vec{\bX}_{{\cal J}_i},\vec{E}_{{\cal J}^c},G_1}
\Big(
\int_{{\cal Y}} \sum_{\bx_i}
P_{Y|\bX_i=\bx_i , \vec{E}_{{\cal J}^c} ,G_1, \vec{\bX}_{{\cal J}_i}}(y)^{\frac{1}{1-s}}
dy \Big)^{1-s}
\nonumber \\
= &
\EE_{\vec{\bX}_{{\cal J}_i},\vec{E}_{{\cal J}^c},G_1}
\Big(
\int_{{\cal Y}} \sum_{\bx_i}
\Big(
q^{-(k_n+k_n')|{\cal J}^c|}
\sum_{ \vec{\bx}_{{\cal J}^c} \in {\cal S}(G_2,\vec{E}_{{\cal J}^c}) }
P_{Y|\bX_i=\bx_i , \vec{\bX}_{{\cal J}^c}=\vec{\bx}_{{\cal J}^c}, \vec{\bX}_{{\cal J}_i}}(y)
\Big)^{\frac{1}{1-s}}
dy \Big)^{1-s}
\nonumber \\
\stackrel{(a)}{\le} & 
\EE_{\vec{\bX}_{{\cal J}_i}}
\Big(
\int_{{\cal Y}} \sum_{\bx_i}
\EE_{\vec{E}_{{\cal J}^c},G_1}
\Big(
q^{-(k_n+k_n')|{\cal J}^c|}
\sum_{ \vec{\bx}_{{\cal J}^c} \in {\cal S}(G_2,\vec{E}_{{\cal J}^c}) }
P_{Y|\bX_i=\bx_i , \vec{\bX}_{{\cal J}^c}=\vec{\bx}_{{\cal J}^c} ,\vec{\bX}_{{\cal J}_i}}(y)
\Big)^{\frac{1}{1-s}}
dy \Big)^{1-s}
\nonumber \\
= &
\EE_{\vec{\bX}_{{\cal J}_i}}
\Big(
\int_{{\cal Y}} \sum_{\bx_i}
q^{-\frac{(k_n+k_n')}{1-s}|{\cal J}^c|}
\EE_{\vec{E}_{{\cal J}^c},G_1}
\Big(
\sum_{ \vec{\bx}_{{\cal J}^c} \in {\cal S}(G_2,\vec{E}_{{\cal J}^c}) }
P_{Y|\bX_i=\bx_i , \vec{\bX}_{{\cal J}^c}=\vec{\bx}_{{\cal J}^c} ,\vec{\bX}_{{\cal J}_i}}(y)
\nonumber \\
 & \cdot
\Big(
P_{Y|\bX_i=\bx_i , \vec{\bX}_{{\cal J}^c}=\vec{\bx}_{{\cal J}^c} ,\vec{\bX}_{{\cal J}_i}}(y)
+\sum_{ \vec{\bx}_{{\cal J}^c}':{\rm (C1)(C2)(C3)} }
P_{Y|\bX_i=\bx_i , \vec{\bX}_{{\cal J}^c}=\vec{\bx}_{{\cal J}^c}',\vec{\bX}_{{\cal J}_i}}(y)
\nonumber \\
& +\sum_{ \vec{\bx}_{{\cal J}^c}':{\rm (C1)(C2)(C4)} }
P_{Y|\bX_i=\bx_i , \vec{\bX}_{{\cal J}^c}=\vec{\bx}_{{\cal J}^c}',\vec{\bX}_{{\cal J}_i}}(y)
\Big)^{\frac{s}{1-s}}
\Big)
dy \Big)^{1-s}
\nonumber \\
= &
\EE_{\vec{\bX}_{{\cal J}_i}}
\Big(
\int_{{\cal Y}} \sum_{\bx_i}
q^{-\frac{(k_n+k_n')}{1-s}|{\cal J}^c|}
\Big(
q^{-(nl-(k_n+k_n'))|{\cal J}^c|}
\sum_{ \vec{\bx}_{{\cal J}^c}}
P_{Y|\bX_i=\bx_i , \vec{\bX}_{{\cal J}^c}=\vec{\bx}_{{\cal J}^c} ,\vec{\bX}_{{\cal J}_i}}(y)
\nonumber \\
 & \cdot
\EE_{G_1}
\Big(
P_{Y|\bX_i=\bx_i , \vec{\bX}_{{\cal J}^c}=\vec{\bx}_{{\cal J}^c} ,\vec{\bX}_{{\cal J}_i}}(y)
+\sum_{ \vec{\bx}_{{\cal J}^c}': {\rm (C1)(C2)(C3)}}
P_{Y|\bX_i=\bx_i , \vec{\bX}_{{\cal J}^c}=\vec{\bx}_{{\cal J}^c}' ,\vec{\bX}_{{\cal J}_i}}(y)
\nonumber \\ &
+\sum_{ \vec{\bx}_{{\cal J}^c}':{\rm (C1)(C2)(C4)}}
P_{Y|\bX_i=\bx_i , \vec{\bX}_{{\cal J}^c}=\vec{\bx}_{{\cal J}^c}' ,\vec{\bX}_{{\cal J}_i}}(y)
\Big)^{\frac{s}{1-s}}
\Big)
dy \Big)^{1-s}
\nonumber \\
\stackrel{(b)}{\le} & 
\EE_{\vec{\bX}_{{\cal J}_i}}
\Big(
\int_{{\cal Y}} \sum_{\bx_i}
q^{-\frac{(k_n+k_n')}{1-s}|{\cal J}^c|}
\Big(
q^{-(nl-(k_n+k_n'))|{\cal J}^c|}
\sum_{ \vec{\bx}_{{\cal J}^c}}
P_{Y|\bX_i=\bx_i , \vec{\bX}_{{\cal J}^c}=\vec{\bx}_{{\cal J}^c} ,\vec{\bX}_{{\cal J}_i}}(y)
\nonumber \\
 & \cdot
\Big(
\EE_{G_1}\Big(
P_{Y|\bX_i=\bx_i , \vec{\bX}_{{\cal J}^c}=\vec{\bx}_{{\cal J}^c} ,\vec{\bX}_{{\cal J}_i}}(y)
+\sum_{ \vec{\bx}_{{\cal J}^c}': {\rm (C1)(C2)(C3)} }
P_{Y|\bX_i=\bx_i , \vec{\bX}_{{\cal J}^c}=\vec{\bx}_{{\cal J}^c}' ,\vec{\bX}_{{\cal J}_i}}(y)
\nonumber \\ &
+\sum_{ \vec{\bx}_{{\cal J}^c}': {\rm (C1)(C2)(C4)} }
P_{Y|\bX_i=\bx_i , \vec{\bX}_{{\cal J}^c}=\vec{\bx}_{{\cal J}^c}' ,\vec{\bX}_{{\cal J}_i}}(y)
\Big)
\Big)^{\frac{s}{1-s}}
\Big)
dy \Big)^{1-s}
\nonumber \\
\stackrel{(c)}{\le} & 
\EE_{\vec{\bX}_{{\cal J}_i}}
\Big(
\int_{{\cal Y}} \sum_{\bx_i}
q^{-\frac{(k_n+k_n')}{1-s}|{\cal J}^c|}
\Big(
q^{-(nl-(k_n+k_n'))|{\cal J}^c|}
\sum_{ \vec{\bx}_{{\cal J}^c}}
P_{Y|\bX_i=\bx_i , \vec{\bX}_{{\cal J}^c}=\vec{\bx}_{{\cal J}^c} ,\vec{\bX}_{{\cal J}_i}}(y)
\nonumber \\
 & \cdot
\Big(
P_{Y|\bX_i=\bx_i , \vec{\bX}_{{\cal J}^c}=\vec{\bx}_{{\cal J}^c} ,\vec{\bX}_{{\cal J}_i}}(y)
+q^{(k_n+k_n')|{\cal J}^c|}
P_{Y|\bX_i=\bx_i ,\vec{\bX}_{{\cal J}_i}}(y)
\nonumber \\ &
+q^{(k_n+k_n')}
P_{Y|\bX_i=\bx_i , (\bX_j-\bX_{j'}=\bx_j-\bx_{j'})_{j,j' \in {\cal J}^c} ,\vec{\bX}_{{\cal J}_i}}(y)
\Big)^{\frac{s}{1-s}}
\Big)
dy \Big)^{1-s}
\nonumber \\
\stackrel{(d)}{\le} & 
\EE_{\vec{\bX}_{{\cal J}_i}}
\Big(
\int_{{\cal Y}} \sum_{\bx_i}
q^{(-nl-\frac{s}{1-s}(k_n+k_n'))|{\cal J}^c|}
\sum_{ \vec{\bx}_{{\cal J}^c}}
P_{Y|\bX_i=\bx_i , \vec{\bX}_{{\cal J}^c}=\vec{\bx}_{{\cal J}^c} ,\vec{\bX}_{{\cal J}_i}}(y)
\nonumber \\
 & \cdot
\Big(
P_{Y|\bX_i=\bx_i , \vec{\bX}_{{\cal J}^c}=\vec{\bx}_{{\cal J}^c} ,\vec{\bX}_{{\cal J}_i}}(y)^{\frac{s}{1-s}}
+q^{\frac{s}{1-s}(k_n+k_n')|{\cal J}^c|}
P_{Y|\bX_i=\bx_i ,\vec{\bX}_{{\cal J}_i}}(y)^{\frac{s}{1-s}}
\nonumber \\ &
+q^{\frac{s}{1-s}(k_n+k_n')}
P_{Y|\bX_i=\bx_i , (\bX_j-\bX_{j'}=\bx_j-\bx_{j'})_{j,j' \in {\cal J}^c} ,\vec{\bX}_{{\cal J}_i}}(y)^{\frac{s}{1-s}}
\Big)
dy \Big)^{1-s}
\nonumber 
\end{align}

\newpage
\begin{align}
= &
\EE_{\vec{\bX}_{{\cal J}_i}}
\Big(
\int_{{\cal Y}} \sum_{\bx_i}
q^{(-nl-\frac{s}{1-s}(k_n+k_n'))|{\cal J}^c|}
\sum_{ \vec{\bx}_{{\cal J}^c}}
P_{Y|\bX_i=\bx_i , \vec{\bX}_{{\cal J}^c}=\vec{\bx}_{{\cal J}^c} ,\vec{\bX}_{{\cal J}_i}}(y)^{\frac{1}{1-s}}
\nonumber \\
 & +
P_{Y|\bX_i=\bx_i ,\vec{\bX}_{{\cal J}_i}}(y)^{\frac{1}{1-s}}
\nonumber \\ &
+q^{(-\frac{s}{1-s}(k_n+k_n')-nl)(|{\cal J}^c|-1)}
\sum_{ \vec{\bx}_{{\cal J}^c}:(2)}
P_{Y|\bX_i=\bx_i , (\bX_j-\bX_{j'}=\bx_j-\bx_{j'})_{j,j' \in {\cal J}^c} ,\vec{\bX}_{{\cal J}_i}}(y)^{\frac{s}{1-s}}
\Big)
dy \Big)^{1-s}
\nonumber \\
=&
q^{-s (k_n+k_n')|{\cal J}^c|}
e^{s I_{\frac{1}{1-s}}^{\downarrow}(\bY; \bX_i ,\vec{\bX}_{{\cal J}^c} |\vec{\bX}_{{\cal J}_i} ) } 
+e^{s I_{\frac{1}{1-s}}^{\downarrow}
(\bY; \bX_i | \vec{\bX}_{{\cal J}_i} ) } 
+q^{-s(k_n+k_n')(|{\cal J}^c|-1)}
 e^{s I_{\frac{1}{1-s}}^{\downarrow}
(\bY; \bX_i , (\bX_j-\bX_{j'})_{j,j' \in {\cal J}^c}| \vec{\bX}_{{\cal J}_i} ) } ,
\end{align}
where each step can be shown as follows.
$(a)$ follows from the concavity of $x \mapsto x^{1-s}$ with $x\ge 0$.
$(b)$ follows from the concavity of $x \mapsto x^{\frac{s}{1-s}}$ with $x\ge 0$.
$(c)$ follows from the inequalities \eqref{HY11} and \eqref{HY12}. 
$(d)$ follows from the inequality 
$(x+y+z)^{\frac{s}{1-s}}\le 
x^{\frac{s}{1-s}}+y^{\frac{s}{1-s}}+z^{\frac{s}{1-s}}$ with $x,y,z\ge 0$.

\section{Achievability of mutual information rate}\Label{AC}
In this appendix, 
generalizing the method of \cite{HaVa} with $c=2$ to the case with a general positive integer $c$,
we show that the mutual information rate
$I(W)=I(Y;X_{[c]})[W]$ can be achieved in the sense of computation-and-forward
under the code construction presented  in {\bf Step (1)} in Section \ref{S4}.
That is, we show that Receiver can decode 
$F_{[c]}:=\oplus_{j=1}^c F_j=G_3(L_{[c]},M_{[c]})$ with an error probability approaching to zero
when the rate satisfies the condition \eqref{8-20-4}.
In this derivation, we employ basic knowledge for affine maps that is presented in Appendix A of \cite{HaVa}. 

For this aim, we define the degraded channel 
$W_{D}$ from ${\cal X}$ to ${\cal Y}$ as
\begin{align}
W_{D,x}:= \sum_{(x_1, \ldots, x_{c-1})\in \bF_q^{l(c-1)}  }q^{-l(c-1)}
W_{ (x_1, x_2, \ldots, x_{c-1}, x \ominus( \oplus_{i=1}^{c-1} x_i )) }.
\end{align}
Then, we find that $I(W)$ is the mutual information between the input $X$ and the output of 
the degraded channel $W_{D}$ when $X$ is subject to the uniform distribution on $\bF_q^{l}$.
 
In the presented encoder  in {\bf Step (1)} in Section \ref{S4},
the sender $A_i$ encodes the message $F_i$ to 
$G_1(F_i,E_i)$, and  
$\vec{E}=(E_1, \ldots, E_c)$ is subject to the uniform distribution on 
$ \FF_q^{c(nl -(k_n+k_n'))}$.
In the following, we assume that Receiver
makes the decoder dependently only of $G_1$ and $E_{[c]}:= \oplus_{i=1}^{c} E_i$.
Hence, $G_1$ and $E_{[c]}$ are fixed to $g_1$ and $e_{[c]}$.
When we denote $F_{[c]}  \in \bF_q^{k_n+k_n'}$ by $m$,
Receiver's receiving signal is subject to the following distribution; 
\begin{align}
&W^{(n)}_{(g_1(F_1,E_1), \ldots, g_1(F_{c-1},E_{c-1}), g_1(F_c,E_c))  } \nonumber\\
=&W^{(n)}_{(g_1(F_1,E_1), \ldots, g_1(F_{c-1},E_{c-1}), g_1(m \ominus (\oplus_{i=1}^{c-1} F_i) ,e_{[c]}
\ominus (\oplus_{i=1}^{c-1} E_i)))  }\nonumber\\
=&W^{(n)}_{(g_1(F_1,E_1), \ldots, g_1(F_{c-1},E_{c-1}), 
g_1(m ,e_{[c]}) \ominus (\oplus_{i=1}^{c-1} g_1 (F_i, E_i) )  }
\end{align}
Here, $(F_1,E_1), \ldots, (F_{c-1},E_{c-1})$ are independently subject to the uniform distribution 
on $\bF_q^{nl}$, and the decoder to recover $F_{[c]}$ does not depend on them.
Taking the expectation with respect to these variables, we have
\begin{align}
&\mathbb{E}_{(F_1,E_1), \ldots, (F_{c-1},E_{c-1})} 
W^{(n)}_{(g_1(F_1,E_1), \ldots, g_1(F_{c-1},E_{c-1}), 
g_1(m ,e_{[c]}) \ominus (\oplus_{i=1}^{c-1} g_1 (F_i, E_i) )  } \nonumber\\
=& 
W^{(n)}_{D, g_1(m ,e_{[c]}) } .
\end{align}
When Receiver applies the maximum likelihood decoder to the degraded channel $W^{(n)}$,
the randomized choice of
the affine code $G_1(m)+E_{[c]}$ satisfies the condition of Proposition 2 of \cite{HaVa} 
because 
the condition \eqref{HY8} and $G_2 \circ G_1 (F, E)=E$ imply the condition (38) in \cite[Appendix A]{HaVa}.
In the application of Proposition 2 of \cite{HaVa}, 
$G_1$ of this paper corresponds to $F$ in \cite[Appendix A]{HaVa}.
Hence,
the average decoding error probability for $F_{[c]}$
goes to zero
under a rate smaller than $I(W)=I(Y;X_{[c]})[W]$, i.e., 
under the condition \eqref{8-20-4}.




\end{document}